\pdfoutput=1
\RequirePackage{ifpdf}
\ifpdf 
\documentclass[pdftex]{sigma}
\else
\documentclass{sigma}
\fi

\numberwithin{equation}{section}

\newtheorem{Theorem}{Theorem}[section]
\newtheorem{Corollary}[Theorem]{Corollary}
\newtheorem{Lemma}[Theorem]{Lemma}
{ \theoremstyle{definition}
\newtheorem{Definition}[Theorem]{Definition}
\newtheorem{Remark}[Theorem]{Remark} }

\begin{document}

\allowdisplaybreaks

\newcommand{\arXivNumber}{1511.09418}

\renewcommand{\PaperNumber}{034}

\FirstPageHeading

\ShortArticleName{Formal Integrals and Noether Operators}

\ArticleName{Formal Integrals and Noether Operators\\
of Nonlinear Hyperbolic Partial Dif\/ferential Systems\\
Admitting a Rich Set of Symmetries}

\Author{Sergey Ya. STARTSEV}

\AuthorNameForHeading{S.Ya.~Startsev}

\Address{Institute of Mathematics, Ufa Scientific Center, Russian Academy of Sciences, \\ 112 Chernyshevsky Str., Ufa, Russia} 
\Email{\href{mailto:startsev@anrb.ru}{startsev@anrb.ru}}
\URLaddress{\url{http://www.researcherid.com/rid/D-1158-2009}}

\ArticleDates{Received September 16, 2016, in f\/inal form May 18, 2017; Published online May 27, 2017}

\Abstract{The paper is devoted to hyperbolic (generally speaking, non-Lagrangian and nonlinear) partial dif\/ferential systems possessing a full set of dif\/ferential operators that map any function of one independent variable into a symmetry of the corresponding system. We demonstrate that a system has the above property if and only if this system admits a full set of formal integrals (i.e., dif\/ferential operators which map symmetries into integrals of the system). As a consequence, such systems possess both direct and inverse Noether operators (in the terminology of a work by B.~Fuchssteiner and A.S.~Fokas who have used these terms for operators that map cosymmetries into symmetries and perform transformations in the opposite direction). Systems admitting Noether operators are not exhausted by Euler--Lagrange systems and the systems with formal integrals. In particular, a hyperbolic system admits an inverse Noether operator if a dif\/ferential substitution maps this system into a~system possessing an inverse Noether operator.}

\Keywords{Liouville equation; Toda chain; integral; Darboux integrability; higher symmetry; hyperbolic system of partial dif\/ferential equations; conservation laws; Noether theorem}

\Classification{37K05; 37K10; 37K35; 35L65; 35L70}

\section{Introduction}
Let us consider partial dif\/ferential systems of the form
\begin{gather}\label{hyp}
u_{xy}=F(x,y,u,u_x,u_y),
\end{gather}
where $u=(u^1;u^2;\dots;u^n)$ and $F=(F^1;F^2;\dots;F^n)$ are an $n$-dimensional vector and a vector-valued function, $u$ depends on real variables $x$ and $y$. A special class of integrable systems~\eqref{hyp} consists of systems for which there exist $n$ functionally independent $x$- and $y$-integrals (i.e., functions of the forms $w(x,y,u,u_1, \dots, u_k)$, $u_i:=\partial^i u /\partial x^i$ and $\bar{w}(x,y,u,\bar{u}_1, \dots, \bar{u}_m)$, $\bar{u}_i:=\partial^i u /\partial y^i$ such that $D_y(w)=0$ and $D_x(\bar{w})=0$; here $D_y$ and $D_x$ denote the total derivatives with respect to $y$ and $x$ by virtue of the system). Scalar ($n=1$) equations of such kind have been studied in classical works like~\cite{Gurs} since the 19th century as well as in relatively recent papers such as \cite{AK,AYu,Sak,SZh,StM,Tsar,avz,ZhSh,ZhSS,ZhibSok}.

The most known example of scalar equations from the above class is the Liouville equation $u_{xy}=e^u$, for which the minimal-order integrals are
\begin{gather}\label{leint}
\omega=u_{xx}-\frac{u_x^2}{2}, \qquad \bar{\omega}=u_{yy}-\frac{u_y^2}{2}.
\end{gather}
In addition to the integrals, the Liouville equation, according to \cite{ZhSh}, possesses the dif\/ferential operators $\sigma=D_x + u_x$ and $\bar{\sigma}=D_y + u_y$ that respectively map $\ker D_y$ and $\ker D_x$ into solutions $f$ of the linearized Liouville equation $D_x D_y(f) = e^u f$. In other words, $\sigma (g)$ and $\bar{\sigma}(\bar{g})$ are sym\-met\-ries of the Liouville equation for any $g \in \ker D_y$, $\bar{g} \in \ker D_x$. (Generalized symmetries are often written in the form of equations $u_t=f$, but we omit the left-hand sides of these equations for brevity and call $f$ symmetries; i.e., we understand the term `symmetry' as a synonym for `charac\-te\-ris\-tic of a symmetry vector f\/ield' in the terminology of \cite{Olv}.)
Note that all functions of~$x$ and~$y$ respectively lie in $\ker D_y$ and $\ker D_x$ for any system~\eqref{hyp}. In addition, if system \eqref{hyp} admits an $x$-integral $w$ and a $y$-integral $\bar{w}$, then all functions of the forms $g(x,w,D_x(w),D_x^2(w),\dots)$ and $\bar{g}(y,\bar{w}, D_y(\bar{w}),D_y^2(\bar{w}),\dots)$ belong to $\ker D_y$ and $\ker D_x$, respectively. Thus, the symmetry families $\sigma (g)$ and $\bar{\sigma}(\bar{g})$ depend on arbitrary functions.

The Pohlmeyer--Lund--Regge system in the degenerate case
\begin{gather*} u_{xy}^1=\frac{u^2 u^1_x u^1_y}{u^1 u^2 + c},\qquad u_{xy}^2=\frac{u^1 u^2_x u^2_y}{u^1 u^2 + c},\end{gather*}
where $c$ is a constant, can serve as an illustration of the `Liouville-like' properties for $n>1$. The system admits the $x$-integrals
\begin{gather}\label{plri}
w_1=u^1_x u^2_x/\big(u^1 u^2 + c\big), \qquad w_2= u^2_{xx}/u^2_x - \big(u^2 u^1_x\big)/\big(u^1 u^2 + c\big),
\end{gather}
while the `symmetry-driving $x$-operators'
\begin{gather}\label{plrs}
\sigma_1=\left(\begin{matrix} u^1_x \vspace{1mm}\\ u^2_x \end{matrix} \right), \qquad
\sigma_2=\left(\begin{matrix} \dfrac{u^1_x}{w_1} \vspace{1mm}\\ 0 \end{matrix} \right) D_x +
\left(\begin{matrix} -u^1 \vspace{1mm}\\ u^2 \end{matrix} \right)
\end{gather}
map $\ker D_y$ into symmetries of this system. It is convenient to introduce a special term for such operators, and the author of\/fers to use the term `symmetry driver' for them. Replacing $x$ with $y$ in \eqref{plri}, \eqref{plrs}, we obtain $y$-integrals and $y$-symmetry drivers of the system. Additional examples of systems \eqref{hyp} with the same properties can be found, in particular, in \cite{Demsk}.

Moreover, all known examples of system \eqref{hyp} support the conjecture that system \eqref{hyp} admits $n$ both $x$- and $y$-integrals if and only if this system possesses $n$ both $x$- and $y$-symmetry drivers (which, to avoid any assumptions on the existence of integrals in their def\/initions, can be def\/ined as dif\/ferential operators that map arbitrary functions of $x$ and $y$, respectively, into symmetries). In the case $n=1$, the simultaneous existence of integrals and symmetry drivers is well established: results of \cite{AYu,AK,SZh,StM,ZhSS} together prove\footnote{For more details, see Section~1 of~\cite{Starxiv}.} that a scalar equation of the form~\eqref{hyp} has both an $x$-integral and a $y$-integral if and only if there exist dif\/ferential opera\-tors~$\sigma$ and~$\bar{\sigma}$ that map any functions of~$x$ and~$y$, respectively, into symmetries of the equation. The proof is based on the Laplace method of cascade integration (see, e.g.,~\cite{Tr}) that encounters serious dif\/f\/iculties in the case of the systems~($n>1$).

But, by using an alternative way, it was proved in \cite{StS} that system \eqref{hyp} possesses the full set of symmetry drivers if this system admits the full set of $x$- and $y$-integrals. The present paper demonstrates that the converse statement was, in fact, almost proved in \cite{StS}. More precisely, in Section~\ref{efl} we prove that any system \eqref{hyp} with $n$ dif\/ferent both $x$- and $y$-symmetry drivers possesses, in both $x$- and $y$-direction, $n$ dif\/ferent formal integrals (i.e., dif\/ferential operators which map any symmetry of this system into $\ker D_y$ and $\ker D_x$, respectively). We also illustrate by an example that formal integrals allow us to construct the full set of `genuine' integrals from symmetries of \eqref{hyp}. But to guarantee the existence of integrals in the general case, we need, in particular, to prove that formal integrals cannot map all symmetries into functions of $x$ or $y$ only instead of integrals. Such a proof is still absent.

Integrals are a special form of conservation laws, and Noether's theorem is widely used to derive conservation laws from symmetries and for the inverse operation. But the proofs of the statements mentioned in the previous two paragraphs do not use Noether's theorem. Moreover, not all systems admitting symmetry drivers are Euler--Lagrange systems, and this makes the Noether theorem inapplicable (at least in the classical form). On the other hand, the simultaneous existence of integrals and symmetry drivers suggests that some Noether-like relationships may be an underlying hidden reason of this simultaneous existence.

Trying to f\/ind such relationships, one can observe that symmetry drivers allow to construct dif\/ferential operators which map the characteristics of conservation laws into symmetries, and integrals (including formal ones) generate dif\/ferential operators which, roughly speaking, perform transformations in the inverse direction. Extending the terms def\/ined in \cite{ff} for evolution systems to a wider context, it is natural to call such operators Noether and inverse Noether operators, respectively. In Section~\ref{dis} we consider the Noether operators generated by integrals and symmetry drivers, but point out that these Noether operators do not explain the simultaneous existence of integrals and symmetry drivers. Despite the last fact, Noether operators seem to be interesting in themselves and, possibly, may be useful for other purposes. Motivations for this are given in Section~\ref{nomot} and, partially, in Sections~\ref{langs} and~\ref{dis} too.

In Section~\ref{noex} we consider some classes of systems~\eqref{hyp} admitting Noether operators and demonstrate that the existence of Noether operators is not a rare property for systems~\eqref{hyp}. In particular, it is proved in Section~\ref{notr} that system~\eqref{hyp} admits an inverse Noether operator if a~dif\/ferential substitution maps this system into a~system possessing an inverse Noether operator (for example, into an Euler--Lagrange system). Applying this result to the well-known Goursat equation $u_{xy}=\sqrt{u_x u_y}$, we use this equation as an example that illustrates the existence of inverse Noether operators for some systems~\eqref{hyp} possessing neither Lagrangians nor integrals nor symmetry drivers.

\section{Notation and basic def\/initions}
\subsection{Variables, functions and operators} Almost all\footnote{Rare exceptions are either explicitly marked or can be easily recognized from the context.} relations in the present paper are considered on solutions of \eqref{hyp}. To formalize this, we eliminate all mixed derivatives of $u$ from these relations by virtue of system \eqref{hyp} and its dif\/ferential consequences, and then demand that the equalities obtained in this way hold identically in terms of $x$, $y$, $u$ and the rest part of its derivatives. Therefore, we can assume without loss of generality that all local objects (symmetries, $x$- and $y$-integrals, coef\/f\/icients of dif\/ferential operators etc.) are functions of a f\/inite number of variables $x$, $y$, $u$, $u_i:=\partial^i u /\partial x^i$, $\bar{u}_j:=\partial^j u /\partial y^j$. We use the notation $g[u]$ to emphasize that a function $g$ depends on a f\/inite number of the above variables. If a function $g$ may depend on mixed derivatives of $u$ in addition to the above variables (i.e., the exclusion of mixed derivatives is not performed), then we use the notation $g\{u\}$. Our considerations are local, the right-hand side of \eqref{hyp} and all other functions are assumed to be locally analytical.

We employ the term `vector' for column vectors only, while the word `row' is used for row vectors. In particular, $u$ and its derivatives are column vectors. Let the superscript $\top$ denote the transposition operation. If $g$ is a scalar function and $z$ is a vector $(z^{1}, z^{2}, \dots, {z^{\kappa})}^{\top}$, then by $g_{z}={{\partial g}/{\partial z}}$ we denote the row $\big( {{\partial g}/{\partial z^{1}}}, {{\partial g}/{\partial z^{2}}},\dots, {{\partial g}/ {\partial z^{\kappa}}} \big)$. For any vector-valued function $G=\big(G^{1}, G^{2}, \dots, G^{\ell}\big)^{\top}$, $G_z= {{\partial G}/{\partial z}}$ designates the $\ell \times \kappa$ matrix with the rows $G^1_z, \dots, G^\ell_z$. For a function $g[u]$ we use the notation $\operatorname{ord}_x(g) = k$ and $\operatorname{ord}_y(g) = m$ if $k$ and $m$ are the highest integers for which the function $g[u]$ essentially depends on $u_k$ and $\bar{u}_m$, respectively. If a function~$g[u]$ does not depend on $u_k$ or $\bar{u}_m$ for all positive integer $k$ and $m$, then we set $\operatorname{ord}_x(g) =0$ or $\operatorname{ord}_y(g) = 0$, respectively.

Let $D_x$ and $D_y$ denote the operators of total derivatives by virtue of system \eqref{hyp}. For a~function $g[u]$ they are def\/ined by the formulas
\begin{gather*} D_{x}(g) = \frac{\partial g}{\partial x} + \frac{\partial g}{\partial u} u_1 + \sum^{\infty}_{i=1} \left( \frac{\partial g}{\partial u_i} u_{i+1} +\frac{\partial g}{\partial \bar{u}_i} D^{i-1}_y (F) \right),\\
D_{y}(g) = \frac{\partial g}{\partial y} + \frac{\partial g}{\partial u} \bar{u}_1 + \sum^{\infty}_{i=1} \left( \frac{\partial g}{\partial \bar{u}_i} \bar{u}_{i+1}
+\frac{\partial g}{\partial u_i} D^{i-1}_x (F) \right).\end{gather*}
The action of $D_x$ and $D_y$ on vectors and matrices is def\/ined componentwise. To make the notation more compact, in the above formulas and below we set the zero power of any dif\/ferentiation equal to the operator of multiplication by unit (that is, equal to the identity mapping).

Consider a dif\/ferential operator $\mathfrak{S}$
of the form
\begin{gather}\label{do}
\sum _{i=0}^{k} \sum _{j=0}^{m} \xi_{ij}[u] D_x^i D_y^j,
\end{gather}
where $\xi_{ij}[u]$ are matrices. By $\mathfrak{S}^\dagger$ we denote the \emph{formal adjoint} of $\mathfrak{S}$, i.e.,
\begin{gather*} \mathfrak{S}^\dagger =\sum
_{i=0}^{k} \sum _{j=0}^{m} (-1)^{i+j} D_x^i D_y^j \circ
\xi_{ij}^{\top}, \end{gather*}
where the symbol~$\circ$ denotes the composition of operators. Further on, we use the following property of the formal adjoint: $(P \circ Q)^{\dagger} = Q^{\dagger} \circ P^{\dagger}$ for any two dif\/ferential operators~$P$ and~$Q$ of the form~\eqref{do}. This and other operator equalities should be understood as follows: two objects are considered equal if their coef\/f\/icients $\xi_{ij}$ coincide after the objects are transformed to the canonical form \eqref{do}.

\subsection{Symmetries}
\begin{Definition} An $n$-component vector function $f[u]$ is said to be a \emph{symmetry} of system \eqref{hyp} if the function $f$ satisf\/ies the relation $L(f) = 0$, where
\begin{gather}\label{lin}
L = D_x D_y - F_{u_x} D_x - F_{u_y} D_y - F_u.
\end{gather}
\end{Definition}

The dif\/ferential operator \eqref{lin} is called the \emph{linearization operator} of system \eqref{hyp}.

\begin{Definition}\label{lis}
A dif\/ferential operator
\begin{gather}\label{si}
\sigma=\sum_{i=0}^k \varsigma_i[u] D_x^i, \qquad \varsigma_k \ne 0, \qquad k \ge 0,
\end{gather}
where $\varsigma_i$ are $n$-dimensional vectors, is said to be an \emph{$x$-symmetry driver} of system \eqref{hyp} if $\sigma(g(x))$ is a symmetry of this system for any scalar function $g(x)$. In this case, the symmetry family $\sigma(g(x))$ is called a \emph{Liouville-like $x$-symmetry} of~\eqref{hyp}, and the vector $\varsigma_k$ is called the \emph{separant} of the symmetry driver $\sigma$ and the corresponding Liouville-like symmetry $\sigma(g(x))$. We say that symmetry drivers $\sigma_1, \sigma_2,\dots, \sigma_r$ and the corresponding Liouville-like symmetries are \emph{essentially independent} if the $n \times r$ matrix consisting of their separants has rank~$r$.
\end{Definition}

If we replace $x$ with $y$ in the above def\/inition, then we obtain the def\/inition of $y$-symmetry drivers and Liouville-like $y$-symmetries. Using the symmetry of formula \eqref{hyp} with respect to the interchange $x \leftrightarrow y$, we hereafter give only one of two `symmetric' def\/initions and statements.
\begin{Lemma}\label{ops}
An operator $\sigma$ is an $x$-symmetry driver if and only if $\sigma$ has the form \eqref{si} and satisfies the operator equality $L \circ \sigma = \varrho \circ D_y$, where $\varrho= (D_x-F_{u_y}) \circ \sigma$.
\end{Lemma}
\begin{proof}
Because all functions of $x$ belong to $\ker D_y$, the equality $L \circ \sigma = \varrho \circ D_y$ directly implies $\sigma(g) \in \ker L$ for any $g(x)$. Therefore, we need to prove the converse only.

For any $\sigma$ of the form \eqref{si}, it is easy to calculate that
\begin{gather*}L \circ \sigma = (D_y-F_{u_x})(\varsigma_k) D_x^{k+1} + \sum _{i=1}^{k}
\left( L(\varsigma_i) + (D_y-F_{u_x}) (\varsigma_{i-1})
\right) D_x^i + L(\varsigma_0)+\varrho \circ D_y,\end{gather*}
where $\varrho= (D_x-F_{u_y}) \circ \sigma$. Applying both sides of this equality to $g(x)$ and taking $L(\sigma(g))=0$ and the arbitrariness of $g$ into account, we obtain the chain of the relations
\begin{gather}
(D_y-F_{u_x})(\varsigma_k) =0, \nonumber\\
L(\varsigma_i) + (D_{y}-F_{u_x}) (\varsigma_{i-1}) = 0,\quad i=1,\dots,k, \label{lattr}\\
L(\varsigma_0) =0.\nonumber
\end{gather}
Since the left-hand sides of~\eqref{lattr} are the coef\/f\/icients at $D_x^i$ in $L \circ \sigma$, their vanishing implies $L \circ \sigma = \varrho \circ D_y$.
\end{proof}

Formulas \eqref{leint}, \eqref{plri} in the introduction illustrate that for some systems~\eqref{hyp} the kernel of~$D_y$ may contain not only functions of~$x$. And the equality $L \circ \sigma = \varrho \circ D_y$ implies that $\sigma(g)$ is a~symmetry of \eqref{hyp} for any $g[u] \in \ker D_y$. It is natural to call such symmetries Liouville-like too, but this is not included in Def\/inition~\ref{lis} to emphasize the absence of any assumption about the structure of $\ker D_y$ in this def\/inition.

If \eqref{hyp} admits $x$-symmetry drivers $\sigma_1, \sigma_2, \dots, \sigma_r$, then $\sigma_1(g^1)+\sigma_2(g^2) + \cdots + \sigma_r(g^r)$ is a~symmetry for any vector $g=(g^1(x),g^2(x), \dots,g^r(x))^\top$. In other words, this set of $x$-symmetry drivers can be considered as one operator $S=\sum_{i=0}^k \alpha_i [u] D_x^i$ such that $\alpha_i[u]$ are $n \times r$ matrices and $S$ maps any $r$-dimensional vector depending on $x$ (and other elements of $\ker D_y$ if such elements exist) into a symmetry of system \eqref{hyp}. If $\sigma$ is an $x$-symmetry driver, then $\sigma \circ D_x^i$ for any $i>0$ is an $x$-symmetry driver too. (It is easy to see, for example, form the equality $L \circ \sigma = \varrho \circ D_y$.) Therefore, we always can equalize the highest powers of $D_x$ in $\sigma_1, \sigma_2, \dots, \sigma_r$ and construct the corresponding operator $S$ so that its separant $\alpha_k$ has rank~$r$ if $\sigma_1, \sigma_2, \dots, \sigma_r$ are essentially independent. For example, the $x$-symmetry drivers $\sigma_1$, $\sigma_2$ def\/ined by~\eqref{plrs} can be represented as one operator
\begin{gather*} S= \left(\begin{matrix} u^1_x & \dfrac{u^1_x}{w_1} \\ u^2_x & 0 \end{matrix} \right) D_x +
\left(\begin{matrix} 0 & -u^1 \\ 0 & u^2 \end{matrix} \right),\end{gather*}
where $w_1$ is def\/ined by \eqref{plri}.

\subsection{Linearizations and integrals}
For any function $g[u]$, we def\/ine the dif\/ferential operator
\begin{gather*} g_{*}=\frac{\partial g}{\partial u} + \sum^{\infty}_{i=1}
\left( \frac{\partial g}{\partial \bar{u}_i} D^i_y +
\frac{\partial g}{\partial u_i} D^{i}_{x} \right) \end{gather*}
and call it the \emph{linearization} of the function $g$. The linearization of a vector function $g[u]$ is def\/ined by the same formula. For any $n$-dimensional vector-valued function $f[u]$ we can consider the dif\/ferentiation $\partial_f$ with respect to $t$ by virtue of the equation $u_t=f[u]$. It is easy to see that $\partial_f(g) = g_*(f)$. The direct calculation (see, for example, \cite{fi}) shows that
\begin{gather}\label{dlin}
D_y \circ g_* - \left( D_y(g) \right)_{*} = \sum_{i=0}^p \gamma_i[u] D_x^i \circ L, \qquad D_x \circ g_* - \left( D_x(g) \right)_{*} = \sum_{i=0}^{q} \bar{\gamma}_i[u] D_y^i \circ L,
\end{gather}
where $p=\max(0,\operatorname{ord}_x(g)-1)$, $q=\max(0,\operatorname{ord}_y(g)-1)$, $L$ is def\/ined by \eqref{lin}, and $\gamma_i[u]$, $\bar{\gamma}_i[u]$ are $n$-dimensional rows (or $\ell \times n$ matrices if $g$ is an $\ell$-dimensional vector). Equations \eqref{dlin} are a~more formal version of the well-known fact that $D_x$, $D_y$ commute with $\partial_f$ if $f[u]$ is a symmetry of \eqref{hyp}. ($[\partial _f,D_y] = [\partial _f,D_x] =0$ directly follows from \eqref{dlin} and $\partial_f(g) = g_*(f)$.)

\begin{Definition} A function $w[u]$ is called an \emph{$x$-integral} of system \eqref{hyp} if $D_y(w) = 0$. If $w$ depends on $x$ only, then $w$ is called a trivial $x$-integral.
\end{Definition}
For brevity, below we use the term `integral' as a synonym for `non-trivial integral'. It is easy to see that $x$-integrals cannot depend on the variables $\bar{u}_i$. The order of the highest partial derivative $u_i$ on which an $x$-integral essentially depends is called the order of this integral. Let $w^1,w^2, \dots, w^r$ be $x$-integrals of orders $p_1, p_2, \dots, p_r$, respectively. We say that these integrals are \emph{essentially independent} if the $r \times n$ matrix composed of the rows ${\partial w^1}/{\partial u_{p_1}}, {\partial w^2}/{\partial u_{p_2}}$, $\dots,{\partial w^r}/{\partial u_{p_r}}$ has rank~$r$. It is clear that system \eqref{hyp} cannot have more than~$n$ essentially independent $x$-integrals. An $n$-dimensional system of the form~\eqref{hyp} is called \emph{Darboux integrable} if this system admits both~$n$ essentially independent $x$-integrals and~$n$ essentially independent $y$-integrals.

\begin{Remark}\label{nondi}
There exist systems that admit integrals but are not Darboux integrable (i.e., the number of their essentially independent integrals is less than $2n$). Such systems also have interesting properties. For example, only one integral is suf\/f\/icient to construct an inverse Noether operator (see Remark~\ref{pno} below). If~\eqref{hyp} is an Euler--Lagrange system, then we also need only one integral to guarantee the existence of a symmetry driver for this system (see Section~\ref{langs}).

In addition, $n$ essentially independent $x$-integrals of \eqref{hyp} def\/ine a dif\/ferential substitution for any system $u_t=f(x,u,u_x,\dots,u_k)$ such that $f$ is a symmetry of \eqref{hyp}. This fact was noted for scalar equations in \cite{Sokumn,ZhSS} and used for systems~\eqref{hyp} in \cite{Demsk,Kis}. No $y$-integrals are necessary for this purpose if we are sure that a symmetry of the form $f(x,u,u_x,\dots,u_k)$ exists (otherwise, to guarantee the existence of such symmetries, we need additional assumptions such as the existence of a Lagrangian for~\eqref{hyp} or employing $y$-integrals in the way described at the beginning of Section~\ref{efl}).
\end{Remark}

If $w[u]$ and $\bar{w}[u]$ respectively are $x$- and $y$-integrals, then we obtain
\begin{gather}\label{fint}
D_y \circ w_* = \sum_{i=0}^p \gamma_i[u] D_x^i \circ L, \qquad D_x \circ \bar{w}_* = \sum_{i=0}^q \bar{\gamma}_i[u] D_y^i \circ L,
\end{gather}
by applying \eqref{dlin} to the def\/ining relations $D_y(w)=0$, $D_x(\bar{w})=0$. The following def\/inition can be considered as a generalization of \eqref{fint}.

\begin{Definition}\label{formi} A dif\/ferential operator $\wp = \sum\limits_{i=0}^{p+1} \nu_i [u] D_x^i$, where $\nu_i$ are $n$-dimensional rows, $p \ge 0$, $\nu_{p+1} \ne 0$, is called a \emph{formal $x$-integral} of order $p+1$ for~\eqref{hyp} if the operator identity
\begin{gather}\label{fintc}
D_y \circ \wp = \mathfrak{S} \circ L
\end{gather}
holds for an operator $\mathfrak{S}$ of the form \eqref{do}. We call $\nu_{p+1}$ \emph{separant} of $\wp$ and say that
formal integrals $\wp_1$, $\wp_2$, $\dots$, $\wp_r$ are \emph{essentially independent} if the $r \times n$ matrix composed of their separants has rank $r$.
\end{Definition}
Equation \eqref{fintc} means that $\wp$ maps symmetries (if they exist) into $\ker D_y$. This can be considered as an alternative def\/inition of formal integrals but requires the additional assumption on the symmetry existence. Therefore, we use \eqref{fintc} as a def\/ining relation instead.

Since the left-hand side of \eqref{fintc} does not contain terms with $D_y^j$, $j>1$, the operator $\mathfrak{S}$ cannot contain non-zero powers of $D_y$. We can rewrite \eqref{fintc} as
\begin{gather*} \wp \circ D_y + \dots = \mathfrak{S} \circ (D_x - F_{u_y}) \circ D_y + \cdots, \end{gather*}
where the dots denote the terms without $D_y$. Therefore,
\begin{gather}\label{fact}
\wp = \mathfrak{S} \circ (D_x - F_{u_y}).
\end{gather}
Thus, $\mathfrak{S}$ in \eqref{fintc} contains terms with powers of $D_x$ only and the highest of these powers is $p$.

Equation \eqref{fintc} implies that $D_x^j \circ \wp$, $\forall\, j>0$, is a formal $x$-integral too. Therefore, if $\wp_1, \wp_2$, $\dots, \wp_r$ are essentially independent formal $x$-integrals of orders less or equal to $\bar{p}+1$, then we can construct essentially independent formal $x$-integrals $\tilde{\wp}_1, \tilde{\wp}_2, \dots, \tilde{\wp}_r$ of order $\bar{p}+1$. The vector $(\tilde{\wp}_1, \tilde{\wp}_2, \dots, \tilde{\wp}_r)^\top$ can be considered as one operator $\Omega = \sum\limits_{i=0}^{\bar{p}+1} \beta_i [u] D_x^i$ such that $\beta_i$ are $r \times n$ matrices, $\operatorname{rank} (\beta_{\bar{p}+1}) = r$ and the equality $D_y \circ \Omega = \sum\limits_{i=0}^{\bar{p}} \gamma_i[u] D_x^i \circ L$ holds for some $r \times n$ matri\-ces~$\gamma_i$.

Equalities \eqref{fint} mean that $w_*$ and $\bar{w}_*$ are formal $x$- and $y$-integrals if $w$ and $\bar{w}$ are $x$- and $y$-integrals, respectively. In addition, formal integrals can be derived from symmetry drivers in some situations. For example, the system
\begin{gather}\label{toda}
u_{xy}^1= \exp\big(2 u^1 - u^2\big), \qquad u_{xy}^2= \exp\big(2 u^2 - u^1\big)
\end{gather}
admits the $x$-symmetry drivers
\begin{gather}
\sigma _1 = \left( \begin{matrix} 1 \\ 1 \end{matrix} \right) D_x + \left( \begin{matrix} u_x^1 \vspace{1mm}\\ u_x^2 \end{matrix} \right), \!\qquad
\sigma _2 = \left( \begin{matrix} 2 \\ 1 \end{matrix} \right) D_x^2 +
\left( \begin{matrix} 3 u^1_x \\ 0 \end{matrix} \right) D_x +
\left( \begin{matrix} \vartheta^2_{xx} - u^1_x \vartheta^2_x + 2 u_x^2 \vartheta^1_x \vspace{1mm}\\ - \vartheta^1_{xx} + u^2_x \vartheta^1_x - 2 u_x^1 \vartheta^2_x \end{matrix} \right),\!\!\!\label{todsi}
\end{gather}
where $\vartheta^1=2 u^1 - u^2$, $\vartheta^2=2 u^2 - u^1$. These symmetry drivers were obtained in \cite{LSh,LSSh} by the so-called descent method without using any information about integrals of \eqref{toda}. According to Lemma~\ref{ops}, the operators $\sigma_1$, $\sigma_2$ satisfy the equalities $L \circ \sigma_i = D_x \circ \sigma_i \circ D_y$, $i=1,2$. On the other hand, it is easy to see that the linearization operator \eqref{lin} of system \eqref{toda} satisf\/ies the equality $\mathcal{N} \circ L = L^\dagger \circ \mathcal{N}$, where
\begin{gather*} \mathcal{N} = \left( \begin{matrix} 2 & -1 \\ -1 & 2 \end{matrix} \right). \end{gather*}
Therefore, $D_y \circ \sigma_i^\dagger \circ D_x \circ \mathcal{N} = \sigma_i^\dagger \circ L^\dagger \circ \mathcal{N} = \sigma_i^\dagger \circ \mathcal{N} \circ L$ and
\begin{gather}
\wp_1 = - \sigma_1^\dagger \circ D_x \circ \mathcal{N} = (1,1) D_x^2 - \big(\vartheta^1_x, \vartheta^2_x \big) D_x, \nonumber\\
\wp_2 = \frac{1}{3} \sigma_2^\dagger \circ D_x \circ \mathcal{N} =
(1,0) D_x^3 - u^1_x (2,-1) D_x^2 - ( \vartheta^1_{xx} - u_x^2 \vartheta^1_x, u_x^1 \vartheta^2_x ) D_x\label{todfi}
\end{gather}
are formal $x$-integrals of \eqref{toda}. These formal integrals are essentially independent, and the corresponding operator $\Omega$ for them is
\begin{gather*} \left( \begin{matrix} D_x \circ \wp_1 \\ \wp_2 \end{matrix} \right) =
\left( \begin{matrix} 1 & 1 \\ 1 & 0 \end{matrix} \right) D_x^3 -
\left( \begin{matrix} \vartheta^1_x & \vartheta^2_x \vspace{1mm}\\ 2 u^1_x & - u^1_x \end{matrix} \right) D_x^2
- \left( \begin{matrix} \vartheta^1_{xx} & \vartheta^2_{xx} \vspace{1mm}\\
\vartheta^1_{xx} - u_x^2 \vartheta^1_x & u_x^1 \vartheta^2_x \end{matrix} \right) D_x. \end{gather*}
To derive the above formal integrals from the symmetry drivers, we use the special property of system \eqref{toda} (see Section~\ref{langs} for more details). But special properties are not necessary for the construction of formal integrals. In the next section we demonstrate this by proving that the existence of $n$ symmetry drivers in both $x$- and $y$-direction always implies the existence of formal integrals.

\section{Existence of formal integrals}\label{efl}
The following proposition was proved in \cite{StS}. If system \eqref{hyp} possesses $n$ essentially independent $x$-integrals and $n$ essentially independent $y$-integrals then there exist $n \times n$ matrices $\alpha_i[u]$, $i=\overline{0,k}$, and $\bar{\alpha}_j [u]$, $j=\overline{0,m}$, such that $\det(\alpha_k) \ne 0$, $\det(\bar{\alpha}_m) \ne 0$ and the operators
\begin{gather}\label{sop}
S = \sum_{i=0}^k \alpha_i [u] D_x^i, \qquad \bar{S} = \sum_{i=0}^m \bar{\alpha}_i [u] D_y^i
\end{gather}
map any $n$-dimensional vector respectively composed of elements from $\ker D_y$ and $\ker D_x$ into symmetries of system \eqref{hyp}. But, to prove the above proposition, the work \cite{StS} in fact uses a~more weak condition instead of the existence of the integrals.

Namely, the proof deals with only the linearizations of the def\/ining relations $D_y(w)=0$, $D_x(\bar{w})=0$, which have the form \eqref{fint}. For brevity, here by $w$ and $\bar{w}$ we denote $n$-dimensional vectors composed of the essentially independent $x$-integrals and $y$-integrals, respectively. If we replace $w_*$, $\bar{w}_*$ in \eqref{fint} with dif\/ferential operators
\begin{gather}\label{omg}
\Omega = \sum_{i=0}^{p+1} \beta_i [u] D_x^i, \qquad \bar{\Omega}= \sum_{i=0}^{q+1} \bar{\beta}_i [u] D_y^i,
\end{gather}
then all reasonings of the proof remain valid. Moreover, the linearization operator $L$ in \eqref{fint} (and in all subsequent reasonings) can also be replaced with any operator of the form
\begin{gather}\label{mop}
M= D_x D_y + A[u] D_x + B[u] D_y + C[u],
\end{gather}
where $A$, $B$ and $C$ are $n\times n$ matrices, and this replacement does not harm the proof in \cite{StS} if we simultaneously replace the term `symmetries' with `elements of $\ker M$'.

Thus, applying the reasonings from \cite{StS} and the proof of Lemma~\ref{ops} to this more general situation, we can prove the following statement. (The proof is omitted because it repeats almost verbatim the proof from \cite{StS}.)
\begin{Theorem}\label{t1}
Let the equalities
\begin{gather}\label{fintm}
D_y \circ \Omega = \hat{S}^{\dagger} \circ M, \qquad D_x \circ \bar{\Omega} = \hat{\bar{S}}^{\dagger} \circ M
\end{gather}
hold for an operator $M$ of the form \eqref{mop}, operators $\Omega$, $\bar{\Omega}$ of the form \eqref{omg} and operators
\begin{gather*} \hat{S} = \sum_{i=0}^p \hat{\alpha}_i[u] D_x^i, \qquad \hat{\bar{S}}= \sum_{i=0}^q \hat{\bar{\alpha}}_i[u] D_y^i, \end{gather*}
where $\beta_i$, $\bar{\beta}_i$, $\hat{\alpha}_i$, $\hat{\bar{\alpha}}_i$ are $n \times n$ matrices and $\det (\beta_{p+1}) \ne 0$, $\det(\bar{\beta}_{q+1}) \ne 0$. Then there exist operators $S$, $\bar{S}$ of the form \eqref{sop} such that $\alpha_i$, $\bar{\alpha}_i$ are $n \times n$ matrices, $\det(\alpha_k) \ne 0$, $\det(\bar{\alpha}_m) \ne 0$ and the relations
\begin{gather}\label{fsym}
M \circ S = - \hat{\Omega}^{\dagger} \circ D_y, \qquad M \circ \bar{S} = - \hat{\bar{\Omega}}^{\dagger} \circ D_x
\end{gather}
hold, where $\hat{\Omega}^{\dagger} = -(D_x + B) \circ S$, $\hat{\bar{\Omega}}^{\dagger} = -(D_y + A) \circ\bar{S}$.
\end{Theorem}

Applying Theorem~\ref{t1} to the formal adjoint of \eqref{fsym}, we easily see that the converse statement is also true. We therefore obtain that, brief\/ly speaking, the existence of formal integrals $\Omega$, $\bar{\Omega}$ is equivalent to the existence of symmetry drivers~$S$,~$\bar{S}$. More accurately, taking Def\/initions~\ref{lis} and~\ref{formi}, and comments after them into account, we can reformulate Theorem~\ref{t1} and its converse for $M=L$ in the following compact form.
\begin{Theorem}\label{tc}
System \eqref{hyp} admits $n$ essentially independent Liouville-like $x$-symmetries and~$n$ essentially independent Liouville-like $y$-symmetries if and only if this system possesses $n$ essentially independent formal $x$-integrals and $n$ essentially independent formal $y$-integrals.
\end{Theorem}

When the present paper was preparing for the publication, it was conjectured in \cite{Levi} that the Darboux integrability is equivalent to the existence of symmetries depending on arbitrary functions. This conjecture was formulated for both partial dif\/ferential and partial dif\/ference equations, the form of which was not explicitly specif\/ied. As it is noted in the introduction, this conjecture was, in fact, already proved for scalar equations \eqref{hyp}. For systems \eqref{hyp}, the above hypothesis almost follows from Theorem~\ref{tc} because the Darboux integrability provides us with formal integrals and we need only to derive the existence of `genuine' integrals from the existence of formal ones. The situation is the same for scalar discrete and semi-discrete analogues of~\eqref{hyp} because discrete and semi-discrete versions of Theorem~\ref{tc} are also valid according to \cite{Starxiv}.

Since formal $x$- and $y$-integrals map symmetries into $\ker D_y$ and $\ker D_x$, respectively, we can try to obtain integrals by applying the formal integrals to a symmetry of system \eqref{hyp}. For example, any autonomous system and, in particular, the system \eqref{toda} admit the symmetry $u_x$. The formal integrals \eqref{todfi} map this symmetry into the essentially independent $x$-integrals
\begin{gather*} \wp_1(u_x)= u^1_{xxx} + u^2_{xxx} - u^1_{xx} \vartheta^1_x - u^2_{xx} \vartheta^2_x, \\
\wp_2(u_x)=u^1_{xxxx} - u^1_x \vartheta^1_{xxx} - u^1_{xx} \vartheta^1_{xx}-u^2_{xx} u^1_x \vartheta^2_x +
u^1_{xx} u^2_x \vartheta^1_x. \end{gather*}
Note that symmetries of \eqref{hyp} always exist together with formal integrals by Theorem~\ref{tc}. In particular, the symmetry $u_x$ of \eqref{toda} can also be obtained by the formula $u_x=\sigma_1(1)$, where $\sigma_1$ is def\/ined by \eqref{todsi}.

The work \cite{ZhibSok} of\/fers another way to derive integrals from symmetry drivers and formal integrals. Let $\sigma$ be an $x$-symmetry driver and $\wp$ be a formal $x$-integral. Then the operator $\wp \circ \sigma$, which can be rewritten as $\sum\limits_{i=0}^j w_i[u] D_x^i$, maps $\ker D_y$ into $\ker D_y$ again. But this is possible only if $w_i \in \ker D_y$. For example, let us consider the composition $\wp_2 \circ \sigma_1$, where $\wp_2$ and $\sigma_1$ are respectively def\/ined by \eqref{todfi} and \eqref{todsi}. The direct calculation gives us
\begin{gather*} \wp_2 \circ \sigma_1 = D_x^4 + w_2 D_x^2 + 3 w_1 D_x + D_x(w_1), \end{gather*}
where
\begin{gather*} w_1=u_{xxx}^1 + u_x^1 \big(u_{xx}^2 - 2 u_{xx}^1\big) + \big(u_x^1\big)^2 u_x^2 - u_x^1 \big(u_x^2\big)^2,\\
 w_2= u_{xx}^1+u_{xx}^2 + u_x^1 u_x^2 - \big(u_x^1\big)^2 - \big(u_x^2\big)^2 \end{gather*}
are essentially independent $x$-integrals of smallest orders for \eqref{toda}.

Thus, in the context of Theorem~\ref{tc} we have enough tools to construct `genuine' integrals for concrete examples of system \eqref{hyp}. But these tools may give us functions of $x$ or $y$ only instead of $x$- or $y$-integrals, respectively. Therefore, in the general situation we need to prove that these tools generate not only functions of $x$ or $y$ but also $n$ essentially independent integrals in each of the characteristics.
Such a proof is still absent.

\section{Conservation laws, cosymmetries and Noether operators}\label{nomot}
\begin{Definition} An equality of the form
\begin{gather}\label{conl}
D_x(a[u])=D_y(b[u]),
\end{gather}
where $a$ and $b$ are scalar functions, is called a \emph{conservation law} of system \eqref{hyp}. The conservation law is said to be \emph{trivial} if there exists a scalar function $c[u]$ such that $a=D_y(c)$, $b=D_x(c)$.
\end{Definition}
\begin{Definition} An $n$-component vector function $g[u]$ is called a \emph{cosymmetry} of system \eqref{hyp} if $L^{\dagger}(g)=0$, where $L$ is def\/ined by formula \eqref{lin}.
\end{Definition}
The cosymmetries are also known as `adjoint symmetries' and `conserved covariants' (see, e.g., \cite{SCC} and \cite{ff}, respectively). Following many works (such as \cite{DemSok,Serg}), the author prefers the term `cosymmetry' because of its briefness.

It is known that a characteristic of a non-trivial conservation law coincides with a non-zero cosymmetry on solutions of \eqref{hyp} (see, for example, \cite[Theorem 4.26, Proposition 5.49 and equation (5.83)]{Olv}). Recall that any conservation law~\eqref{conl} can be represented in the characteristic form
\begin{gather}\label{ccl}
\frac{{\tt d}\tilde{a}}{\tt dx} - \frac{{\tt d}\tilde{b}}{\tt dy} = \tilde{G} \cdot (u_{xy} - F(x,y,u,u_x,u_y)),
\end{gather}
where $\cdot$ denotes the scalar product, the scalar functions $\tilde{a}\{u\}$ and $\tilde{b}\{u\}$ respectively coincide with~$a$ and~$b$ on solutions of \eqref{hyp}, and the $n$-component vector function $\tilde{G}\{u\}$ is called a~\emph{cha\-rac\-teristic} of conservation law~\eqref{conl}. Here $\frac{\tt d}{\tt dx}$ and $\frac{\tt d}{\tt dy}$ denote the total derivatives with respect to~$x$ and~$y$, respectively, i.e.,
\begin{gather*} \frac{{\tt d} h}{\tt dx} =
{\frac{\partial h}{\partial x}}+ \sum^{\infty }_{i=0} \sum^{\infty }_{j=0} {\frac{\partial h}{\partial u_{i,j}}} u_{i+1,j}, \qquad
\frac{{\tt d} h}{\tt dy} =
{\frac{\partial h}{\partial y}}+ \sum^{\infty }_{i=0} \sum^{\infty }_{j=0} {\frac{\partial h}{\partial u_{i,j}}} u_{i,j+1}, \qquad u_{i,j}:=\frac{\partial^{i+j} u}{\partial x^{i} y^{j}} \end{gather*}
for any function $h\{u\}$. It should be noted that \eqref{ccl} is valid for any function $u(x,y)$ (in contrast to most other equations in the present paper, which take into account that $u$ is an arbitrary solution of a system under consideration). Applying the Euler operator (the variational derivative) to both sides of \eqref{ccl} and then restricting our consideration to solutions of~\eqref{hyp}, we have $L^{\dagger}(G)=0$, where $G[u]$ is obtained from $\tilde{G}$ via excluding the mixed derivatives by virtue of \eqref{hyp} ($G =\tilde{G}$ on solutions of \eqref{hyp}). According to \cite[Theorem~4.26]{Olv}, $G=0$ if and only if the corresponding conservation law is trivial. It should be noted that not all cosymmetries are characteristics of conservation laws.

An `alternative' way for deriving a cosymmetry from a conservation law is to linearize both sides of \eqref{conl} with taking \eqref{dlin} into account. This gives us
\begin{gather}\label{lcl}
D_x \circ a_* - D_y \circ b_* = \Gamma \circ L, \qquad \Gamma = \sum_{i=0}^q \bar{\gamma}_i[u] D_y^i + \sum_{i=0}^{p} \gamma_i[u] D_x^i, \qquad p, q \ge 0.
\end{gather}
The formal adjoint of this equation is
\begin{gather*} \left( b_* \right)^\dagger \circ D_y - \left( a_* \right)^\dagger \circ D_x = L^\dagger \circ \Gamma^\dagger,\end{gather*}
and $\Gamma^\dagger$ therefore maps any constant into $\ker L^\dagger$. Thus, $\Gamma^\dagger(1)$ is a cosymmetry of~\eqref{hyp}. An accurate checking shows that $\Gamma^\dagger(1)$ coincides with the characteristic $G$ of the conservation law, but the equality $\Gamma^\dagger(1)=G$ is not used below and we omit its proof.

Conversely, any cosymmetry generates a formal analogue of the linearized conservation \linebreak law~\eqref{lcl}, i.e., allows us to construct an operator equality of the form
\begin{gather}\label{flcl}
D_x \circ \Phi - D_y \circ \Psi = \Gamma \circ L,
\end{gather}
where the form of $\Gamma$ is given in \eqref{lcl},
\begin{gather*} \Phi = \sum_{i=0}^{q+1} \bar{\varphi}_i[u] D_y^i + \sum_{i=1}^{r} \varphi_i[u] D_x^i, \qquad
\Psi = \sum_{i=1}^{\bar{r}} \bar{\psi}_i[u] D_y^i + \sum_{i=0}^{p+1} \psi_i[u] D_x^i, \qquad r,\bar{r} > 0, \end{gather*}
and $\bar{\varphi}_i$, $\varphi_i$, $\bar{\psi}_i$, $\psi_i$ are $n$-dimensional rows. Indeed, the direct calculation shows that
\begin{gather*} L^\dagger \circ g = g D_x D_y + \big(D_x + F_{u_y}^\top\big)(g) D_y + \big(D_y + F_{u_x}^\top\big)(g) D_x + L^{\dagger} (g)\end{gather*}
for any $n$-dimensional vector $g[u]$. The formal adjoint of this equality is
\begin{gather*} g^{\top} L - \big( L^{\dagger} (g) \big)^\top = D_x D_y \circ g^\top - D_y \circ \big( D_x\big(g^\top\big) + g^\top F_{u_y} \big) - D_x \circ \big( D_y\big(g^\top\big) + g^\top F_{u_x} \big)\\
\hphantom{g^{\top} L - \big( L^{\dagger} (g) \big)^\top}{}
 = D_y \circ g^\top (D_x - F_{u_y}) - D_x \circ \big( D_y\big(g^\top\big) + g^\top F_{u_x} \big)\\
 \hphantom{g^{\top} L - \big( L^{\dagger} (g) \big)^\top}{} =
D_x \circ g^\top (D_y - F_{u_x}) - D_y \circ \big( D_x\big(g^\top\big) + g^\top F_{u_y} \big). \end{gather*}
If $L^{\dagger}(g)=0$, then the above equation takes the form~\eqref{flcl}. Applying both sides of this equation to an $n$-dimensional vector~$f[u]$, we obtain
\begin{gather}
g \cdot L (f) - L^{\dagger}(g) \cdot f = D_y \big( g \cdot (D_x - F_{u_y})(f) \big) - D_x \big( f \cdot \big(D_y + F_{u_x}^\top\big)(g) \big)\nonumber \\
\hphantom{g \cdot L (f) - L^{\dagger}(g) \cdot f }{}
 = D_x \big( g \cdot (D_y - F_{u_x})(f) \big) - D_y \big( f \cdot \big(D_x + F_{u_y}^{\top}\big)(g) \big).\label{fgrp}
\end{gather}

Equation \eqref{fgrp} becomes a conservation law if $f$ is a symmetry and $g$ is a cosymmetry of~\eqref{hyp}. This fact in a more general form was mentioned in~\cite{AB,AB2} (while another method for the reconstruction of conservation laws from their characteristics was the main subject of these works). However, \eqref{fgrp} generates conservation laws that may be (and usually are) trivial. This is not surprising because the operator equality \eqref{flcl} is a generalization of \eqref{lcl} and, accordingly, the application of~\eqref{flcl} to a symmetry $f[u]$ is the direct copy of the same operation for \eqref{lcl}. And the latter operation gives rise to the conservation law $D_x(\partial _f (a))= D_y(\partial _f (b))$, which is usually trivial too. But, like the conservation law $D_x(\partial _f (a))= D_y(\partial _f (b))$, even a trivial conservation law obtained via~\eqref{fgrp} may be a dif\/ferential consequence of a non-trivial conservation law, and we can restore this non-trivial conservation law from its trivial consequence.

As an illustration, let us consider the conservation law
\begin{gather}\label{scl}
D_x\left( \frac{u_y^2}{2}\right) + D_y( \cos u)=0
\end{gather}
for the sine-Gordon equation $u_{xy}=\sin u$. The linearization
\begin{gather*} D_x \circ u_y D_y - D_y \circ \sin u = u_y (D_x D_y - \cos u) \end{gather*}
of this conservation law shows that the corresponding cosymmetry is~$u_y$. Taking into account that $u_y$ is a symmetry too and setting $g=f=u_y$ in~\eqref{fgrp}, we obtain the conservation law
\begin{gather}\label{scl2}
D_x (u_y u_{yy}) - D_y(u_y \sin u) = 0,
\end{gather}
which is trivial because $2 u_y u_{yy} = D_y ( u_y^2)$ and $2 u_y \sin u = D_x (u_y^2)$. But, at the same time, $u_y \sin u = - D_y(\cos u)$ and \eqref{scl2} is the dif\/ferential consequence{\samepage
\begin{gather*} D_y \left (D_x\left( \frac{u_y^2}{2}\right) + D_y( \cos u) \right) =0 \end{gather*}
of the non-trivial conservation law~\eqref{scl}.}

Equation \eqref{fgrp} is a particular case of the formula
\begin{gather}\label{fgr}
g \cdot Z (f) - Z^{\dagger}(g) \cdot f = {\tt Div} (\Lambda),
\end{gather}
which can be considered as a less formal def\/ining relation for the adjoint operator and therefore holds for any dif\/ferential operator $Z$ of a very general form (see, for example, \cite{Olv} for more details). As above, equation \eqref{fgr} becomes a conservation law when $g \in \ker Z^\dagger$ and $f \in \ker Z$. It is suf\/f\/icient to know only a `cosymmetry' $g \in \ker Z^\dagger$ (only a `symmetry' $f \in \ker Z$) for obtaining a conservation law from \eqref{fgr} if there exists an operator that maps $\ker Z^\dagger$ into $\ker Z$ ($\ker Z$ into $\ker Z^\dagger$).

\begin{Definition}\label{nod}
An operator $\mathcal{N}$ is called a \emph{Noether operator} (an \emph{inverse Noether operator}) for an operator $Z$ if there exists a dif\/ferential operator $\breve{\mathcal{N}}$ such that $Z \circ \mathcal{N} = \breve{\mathcal{N}} \circ Z^\dagger$
($Z^\dagger \circ \mathcal{N} = \breve{\mathcal{N}} \circ Z$). We say that a Noether operator $\mathcal{N}$ is \emph{trivial} if $\mathcal{N} = \theta \circ Z^\dagger$ ($\mathcal{N} = \theta \circ Z$) for some dif\/ferential operator $\theta$. In particular, an operator $\mathcal{N}$ of the form \eqref{do} is called a Noether operator (an inverse Noether operator) for system \eqref{hyp} if the corresponding def\/ining relation holds when $Z$ is the linearization operator \eqref{lin} of this system.
\end{Definition}

Less formally, Noether operators (inverse Noether operators) map $\ker Z^\dagger$ into $\ker Z$ ($\ker Z$ into $\ker Z^\dagger$), and trivial Noether operators (that obviously always exist) map $\ker Z^\dagger$ ($\ker Z$) to the zero vector. As far as the author knows, this meaning of the terms `Noether operator' and `inverse Noether operator' are not most common. We follow the work \cite{ff}, in which these terms were used for operators that map cosymmetries of evolution systems into symmetries of these systems and perform transformations in the inverse direction, respectively. Below we also use the term \emph{``Noether's operators''} to designate operators from the set that includes both Noether and inverse Noether operators.

In addition to the derivation of conservation laws from \eqref{fgr}, a Noether operator for the li\-nearization operator of a system can be used to obtain symmetries from non-trivial conservation laws of this system because the characteristics of such conservation laws are non-zero cosym\-met\-ries. (An illustration of such application of Noether operators is included in Section~\ref{langs}.)

\begin{Remark}\label{recur}
It is easy to see that the composition of a Noether operator $\bar{\mathcal{N}}$ and an inverse Noether operator $\mathcal{N}$ is a `recursion operator' for $Z$, i.e., $\bar{\mathcal{N}} \circ \mathcal{N}$ maps $\ker Z$ into $\ker Z$. In other words, Noether operators and inverse Noether operators of the same operator $Z$ are inverse up to `recursion operators' of $Z$. Using the same logic, symmetry drivers could also be called inverse formal integrals since, for example, the composition $\mathcal{R}$ of an $x$-symmetry driver of \eqref{hyp} and a~formal $x$-integral of the same system is a~recursion operator of~\eqref{hyp} (i.e., $\mathcal{R}$ maps the kernel of the linearization operator~\eqref{lin} for the system~\eqref{hyp} into $\ker L$ again).
\end{Remark}

\begin{Remark}\label{nonloc}
It should be emphasized that $\breve{\mathcal{N}}$ and $\theta$ in Def\/inition~\ref{nod} are dif\/ferential (i.e., local) operators (because a \emph{formal} representation $\tilde{Z}=\theta \circ Z$ with non-local $\theta$ does not guarantee that $\tilde{Z}(f)=0$ follows from $Z(f)=0$). For example, $D_y$ is a non-trivial Noether operator for $Z=D_x \circ D_y$ despite the formal representation $D_y=\theta \circ Z$, where $\theta = D_x^{-1}$. Less formally, $D_y$ is non-trivial because it does not map all elements of $\ker \left(D_x \circ D_y\right)$ to zero. This ref\/lects the fact that $D_x^{-1} \circ D_x$ is not the identity mapping in reality, and that non-localities need an accurate treatment. Such treatment is beyond the scope of the present paper and we consider local objects only (with a small exception at the beginning of the next section).
\end{Remark}

\section{Some classes of systems with non-trivial Noether operators}\label{noex}
The f\/irst two sentences of Def\/inition~\ref{nod} are formulated in a general way and impose no restrictions on the type of operators (except the locality of $\breve{\mathcal{N}}$ and $\theta$). Therefore, this def\/inition can be applied to operators of a very general form (in particular, dif\/ferentiations other than~$D_x$ and~$D_y$ can be used to construct them). In addition, we can easily adapt Def\/inition~\ref{nod} to def\/ine Noether's operators for partial dif\/ferential systems of any form.

For example, let us consider an evolution system
\begin{gather}\label{evs}
u_t=f(x,u,u_x,\dots,u_k).
\end{gather}
The work \cite{ff} uses the operator equality
\begin{gather}\label{noevs}
(\partial_f - f_*) \circ \mathcal{N} = \mathcal{N} \circ \big(\partial_f + f_*^\dagger\big)
\end{gather}
as a def\/ining relation for a Noether operator $\mathcal{N}$ of \eqref{evs}. By Def\/inition~\ref{nod}, the equality~\eqref{noevs} also means that $\mathcal{N}$ is a Noether operator for the linearization operators $\partial_f - f_*$ of \eqref{evs} and, hence, for the system~\eqref{evs} itself (if the coef\/f\/icients of $\mathcal{N}$ can be expressed in terms of~$x$,~$u$ and its derivatives with respect to~$x$). Note that~\eqref{noevs} uses $\partial_f$ instead of $d/dt$; this means that~\eqref{noevs} is considered on solutions of~\eqref{evs}. Hamiltonian operators gives us the most known class of Noether operators for evolution systems: a~Hamiltonian operator $\mathcal{H}$ is a Noether operator for any system of the form $u_t=\mathcal{H}(\delta g /\delta u)$, where $\delta /\delta u$ denotes the variational derivative and~$g$ is a scalar function of~$x$,~$u$ and derivatives of $u$ with respect to~$x$.

In this section, we focus on systems of the form \eqref{hyp} only. Interpreting $y$ as a `time' va\-riab\-le, we can rewrite a part of systems \eqref{hyp} in the Hamiltonian form with non-local $\mathcal{H}$. But these Hamiltonian operators $\mathcal{H}$ are not Noether operators for the corresponding systems if we consider these systems in their original form \eqref{hyp}. For example, the sine-Gordon equation can be represented as $u_y= - D_x^{-1} (\frac{\delta \cos u}{\delta u})$, but $D_x^{-1}$ does not map the cosymmetry $u_x$ of the original equation $u_{xy}=\sin u$ into a symmetry. This is because the sets of the cosymmetries are not the same for $u_y= - D_x^{-1} (\frac{\delta \cos u}{\delta u})$ and $u_{xy}=\sin u$. (The representations~\eqref{ccl} and the characteris\-tics~$\tilde{G}$ in them are changed when we change the second factor in the right-hand side of~\eqref{ccl} and permit to express~$\partial^i u /\partial y^i$ by using non-localities.) In addition, to the author's best knowledge, all system~\eqref{hyp} admitting Hamiltonian representations are Euler--Lagrange systems and, therefore, have simple Noether operators applicable to the systems~\eqref{hyp} in their original form (see Section~\ref{langs}). Taking the just-mentioned facts into account and omitting details concerning non-localities (see Remark~\ref{nonloc}), we do not consider Hamiltonian systems as a separate case and give only schematic comments about such systems.

If a system of the form~\eqref{evs} possesses a conservation law
\begin{gather*} \partial_f (b(x,u,u_x,\dots,u_m))=D_x(a(x,u,u_x,\dots,u_{m+k-1})),\end{gather*}
then $\delta b /\delta u$ is a characteristics of the conservation laws. Hence, $\delta b /\delta u$ is a cosymmetry of~\eqref{evs}. Applying this to Hamiltonian representations $u_y=\mathcal{H}(\delta g /\delta u)$ of systems~\eqref{hyp} and replacing $D_y$ with $\partial_{\mathcal{H}(\delta g /\delta u)}$ in the def\/ining relation for $x$-integrals, we obtain that $\mathcal{H}(\delta w /\delta u)$ is a symmetry if $w$ is an $x$-integral of \eqref{hyp}. Since $\xi(x) w$ is also an $x$-integral for any $\xi(x)$ and $\frac{\delta }{\delta u} (\xi(x) w) = (w_*)^\dagger(\xi)$, the operator $\mathcal{H}\circ (w_*)^\dagger$ maps any scalar function $\xi(x)$ into a symmetry. The last fact was obtained in the general form and then used for a subclass of so-called sigma models\footnote{An Euler--Lagrange system is called a sigma model if its Lagrangian has the form $\sum\limits_{i=1}^{n} \sum\limits_{j=1}^{n} H_{ij}(u) u_x^i u_y^j + \zeta (u)$.} in \cite{Demsk}, but this work gives no answer why the operators $\mathcal{H}\circ (w_*)^\dagger$ are local. (They are local in all examples considered in \cite{Demsk}.) We give the answer in Section~\ref{langs} by demonstrating that, in the case of the sigma models, these operators coincide with local $x$-symmetry drivers obtained via a formula for more general Euler--Lagrange systems~\eqref{hyp}.

Let $W$ be an $n$-dimensional vector and the components of $W$ be essentially independent $x$-integrals of minimal orders for system \eqref{hyp}. If this system can be rewritten in the Hamiltonian form $u_y=\mathcal{H}(\delta g /\delta u)$, then $\tilde{\mathcal{H}}=W_* \circ \mathcal{H} \circ W_*^\dagger$ is the composition of the formal integral $W_*$ and the $x$-symmetry driver $\mathcal{H} \circ W_*^\dagger$. Therefore (see the penultimate paragraph of Section~\ref{efl}), the coef\/f\/icients of $\tilde{\mathcal{H}}$ belong to $\ker D_y$ and, according to \cite{ZhG,avz,ZMHS}, can be expressed in terms of~$x$,~$W$ and total derivatives of~$W$ with respect to $x$. On the other hand, the formula for $\tilde{\mathcal{H}}$ coincides with the well-known rule for recalculating a Hamiltonian operator $\mathcal{H}$ under a transformation $v=W[u]$ (see, for example,~\cite{ff,KupW}). Thus, taking the last sentence of the previous paragraph into account, we see that the sigma models admitting the full set of $x$-integrals allow us to construct local Hamiltonian operators. This was noted (without establishing the locality of $\tilde{\mathcal{H}}$) in~\cite{Demsk} and proved for a special subclass of the sigma models in~\cite{KisL}.

\subsection{Euler--Lagrange systems}\label{langs}
Any system \eqref{hyp} has a non-trivial Noether operator if this system is derived from a Lagrangian of the form
\begin{gather}\label{lang}
u_x \cdot H(x,y,u) u_y + \mu (x,y,u)u_x + \bar{\mu} (x,y,u) u_y + \zeta (x,y,u),
\end{gather}
where $\zeta$ is a scalar function, $\mu$, $\bar{\mu}$ are $n$-dimensional rows, and $H$ is an $n\times n$ matrix. It is easy to check (see, for example, \cite{intmat}) that linearization operators \eqref{lin} of such systems satisfy the operator identity $\mathcal{N} \circ L = L^\dagger \circ \mathcal{N}$, where $\mathcal{N} = H + H^\top$. Hence, $\mathcal{N}$ is an inverse Noether operator, and, if $\mathcal{N}$ is non-degenerate\footnote{We cannot guarantee the equivalence of the corresponding Euler--Lagrange system to a system of the form~\eqref{hyp} if $\mathcal{N}$ is degenerate.}, $\bar{\mathcal{N}}= (H^\top +H)^{-1}$ is a Noether operator of system \eqref{hyp}. For example, $\mathcal{N}$ is the identity mapping ($\mathcal{N}=1$) for the aforementioned sine-Gordon equation (as well as for any system of the form $u_{xy}=\zeta_u^\top$ generated by \eqref{lang} with $H=1/2$, $\mu=\bar{\mu}=0$).

The operator $\bar{\mathcal{N}}$ allows us, in particular, to construct a symmetry driver from any integral of~\eqref{hyp}. Indeed, if $D_y(w)=0$, then \eqref{dlin} gives rise to $D_y \circ w_* = \mathfrak{S} \circ L$, where $w_*= \mathfrak{S} \circ (D_x - F_{u_y})$ in accordance with \eqref{fact}. But the formal adjoint $L^\dagger \circ \mathfrak{S}^\dagger = - (w_*)^\dagger \circ D_y$ means that $\mathfrak{S}^\dagger$ maps $\ker D_y$ into cosymmetries and $\bar{\mathcal{N}} \mathfrak{S}^\dagger (g[u])$ is therefore a symmetry of \eqref{hyp} for any scalar function $g \in \ker D_y$. Thus,
\begin{gather}\label{nos}
\sigma= \bar{\mathcal{N}} \circ \big(D_x + F_{u_y}^\top\big)^{-1} \circ (w_*)^\dagger
\end{gather}
is a \emph{local} dif\/ferential operator that maps $\ker D_y$ into symmetries. For example, this formula easily gives us the symmetry drivers $\sigma =D_x+u_x$ and $\bar{\sigma}=D_y+u_y$ for the Liouville equation mentioned in the introduction if we apply \eqref{nos} to the integrals \eqref{leint} and take into account that $\bar{\mathcal{N}}=1$ for this equation.

Now let us compare \eqref{nos} and the above mentioned operator $\mathcal{H} \circ (w_*)^\dagger$ in the case of the sigma models. If \eqref{hyp} is an Euler--Lagrange system with a Lagrangian of the form $u_x \cdot H(u) u_y + \zeta (u)$ and $H+H^\top$ is non-degenerate, then \cite{Moh} the operator $\mathcal{H}^{-1}:=(H+H^\top) \left( D_x - F_{u_y}\right)$ is symplectic and the system \eqref{hyp} can be written in the Hamiltonian form $u_y=\mathcal{H}(\delta \zeta /\delta u)$. Since symplectic operators are skew-symmetric, we see that the operator $-\mathcal{H} \circ (w_*)^\dagger$ coincides with the right-hand side of \eqref{nos} and therefore is local.

Note that, in contrast to Theorem~\ref{tc}, it is suf\/f\/icient to have only one integral for constructing the symmetry driver $\sigma$ via formula \eqref{nos}. In particular, the f\/irst and second symmetry drivers~\eqref{plrs} are respectively derived from the f\/irst and second integrals \eqref{plri} by using \eqref{nos} with
\begin{gather*} \bar{\mathcal{N}} = \big(u^1 u^2 + c\big) \left( \begin{matrix} 0 & 1 \\ 1 & 0 \end{matrix} \right). \end{gather*}
This and other applications of \eqref{nos} to systems~\eqref{hyp} can be found in~\cite{fi}. Naturally, \eqref{nos}~can also be derived from the classical Noether theorem in the case of Euler--Lagrange systems (see~\mbox{\cite{KisL, intmat}}). But it is easy to see that~\eqref{nos} is also valid for non-Lagrangian system~\eqref{hyp} if this system admits a Noether operator~$\bar{\mathcal{N}}$ of the form~\eqref{do}.

Not all integrable systems \eqref{hyp} are generated by Lagrangians \eqref{lang}. This is easy to see, in particular, from the existence of integrable systems the right-hand sides $F$ of which depend nonlinearly on $u_x$ or $u_y$. Many integrable equations of such kind can be found in \cite{MeshSok, ZhibSok}. For example, $u_{xy}=u_x \sqrt{u_y}$ is Darboux integrable in accordance with~\cite{MeshSok} and depends nonlinearly on $u_y$. Therefore, no Lagrangian~\eqref{lang} corresponds to this equation.

But systems \eqref{hyp} often possess non-trivial Noether's operators regardless of whether these systems have Lagrangians. In the rest part of the article we consider some of such cases.

\subsection{Darboux integrable systems}\label{dis} Any Darboux integrable system admits non-trivial Noether's operators in four dif\/ferent `directions'. Indeed, equation \eqref{fsym} and its formal adjoint imply
\begin{gather*} M \circ S \circ \hat{\Omega} = - \hat{\Omega}^{\dagger} \circ D_y \circ \hat{\Omega} =- \hat{\Omega}^{\dagger} \circ S^{\dagger} \circ M^{\dagger}, \\
 M \circ \bar{S} \circ \hat{\bar{\Omega}} = - \hat{\bar{\Omega}}^{\dagger} \circ D_x \circ \hat{\bar{\Omega}} = - \hat{\bar{\Omega}}^{\dagger} \circ \bar{S}^{\dagger} \circ M^{\dagger}, \end{gather*}
i.e., the dif\/ferential operators $S\circ \hat{\Omega}$ and $\bar{S} \circ \hat{\bar{\Omega}}$ are Noether operators for $M$. Analogously, equation~\eqref{fintm} and its formal adjoint imply
\begin{gather*} M^{\dagger} \circ \hat{S} \circ \Omega = - \Omega^{\dagger} \circ \hat{S}^{\dagger} \circ M, \qquad
M^{\dagger} \circ \hat{\bar{S}} \circ \bar{\Omega} = - \bar{\Omega}^{\dagger} \circ \hat{\bar{S}}^{\dagger} \circ M, \end{gather*}
and $\hat{S} \circ \Omega$, $\hat{\bar{S}} \circ \bar{\Omega}$ are inverse Noether operators for $M$. These four Noether's operators are non-trivial because they have no terms of the form $\gamma [u] D_x^i D_y^j$, $i \, j \ne 0$, while $M$ contains the term $D_x D_y$. Since both equations \eqref{fintm}, \eqref{fsym} hold true with $M=L$ for any Darboux integrable system \eqref{hyp}, corresponding $S\circ \hat{\Omega}$, $\bar{S} \circ \hat{\bar{\Omega}}$ are Noether operators, and $\hat{S} \circ \Omega$, $\hat{\bar{S}} \circ \bar{\Omega}$ are inverse Noether operators for this system. More generally, Theorem~\ref{t1} and its converse mean that only one of equations \eqref{fintm} and \eqref{fsym} is suf\/f\/icient for $M$ to have all four above Noether's operators.

\begin{Remark}\label{pno}
For brevity, in the previous paragraph we consider $n$ symmetry drivers (as well as $n$ formal integrals) as one operator. But symmetry drivers and formal integrals can be considered separately, their numbers may be less than $n$ and they may exist in only one of $x$- and $y$-directions. Even under these conditions, we can construct a Noether operator (if at least one symmetry driver exists) and an inverse Noether operator (if at least one formal integral exists). For example, if $\sigma_i$, $i=\overline{1,r}$, are $x$-symmetry drivers, then the same reasoning and Lemma~\ref{ops} give us that $\sigma_i \circ \sigma_j^\dagger \circ \big(D_x + F_{u_y}^\top \big)$ are Noether operators for all positive integers $i, j \le r$.
\end{Remark}

It should be noted that above Noether's operators do not explain the simultaneous existence of the formal integrals $\Omega$, $\bar{\Omega}$ and the symmetry drivers $S$, $\bar{S}$ (i.e., Theorems~\ref{t1} and \ref{tc} cannot be proved by using these operators). To obtain symmetries from conservation laws and, in particular, symmetry drivers from \eqref{fintm}, we need a Noether operator (just as a Noether opera\-tor~$\bar{\mathcal{N}}$ is needed for deriving~\eqref{nos}). But~\eqref{fintm} gives us the inverse Noether operators only. The situation for~\eqref{fsym} is the same: we need an inverse Noether operators to obtain formal integrals, but~\eqref{fsym} gives us the Noether operators only. However, these Noether's operators may possibly be useful for some other purposes.

As an example, let us try to construct conservation laws by using \eqref{fgrp} and one of above Noether's operators. Let $f[u], \tilde{f}[u] \in \ker L$ and $g[u] = \mathcal{N} (\tilde{f})$, where $\mathcal{N}=\hat{S} \circ \Omega$ and $\hat{S}$, $\Omega$ satisfy~\eqref{fintm} for~$M=L$. Substituting this into \eqref{fgrp} and taking $\mathcal{N} (\tilde{f}) \in \ker L^\dagger$ into account, we obtain the following conservation law
\begin{gather}\label{nocl}
D_y \big( \hat{S}\big(\Omega\big(\tilde{f}\big)\big) \cdot (D_x - F_{u_y})(f) \big) = D_x \big( f \cdot \big(D_y + F_{u_x}^\top\big)\big(\mathcal{N}\big(\tilde{f}\big)\big) \big).
\end{gather}
Equations \eqref{fgr} and \eqref{fact} give us
\begin{gather*} \hat{S}\big(\Omega\big(\tilde{f}\big)\big) \cdot (D_x - F_{u_y})(f) = \Omega\big(\tilde{f}\big) \cdot \hat{S}^\dagger \big((D_x - F_{u_y})(f)\big) + D_x(c[u]) \\
\hphantom{\hat{S}\big(\Omega\big(\tilde{f}\big)\big) \cdot (D_x - F_{u_y})(f)}{} = \Omega\big(\tilde{f}\big) \cdot \Omega(f) + D_x(c[u]). \end{gather*}
According to \eqref{fintm}, the operator $\Omega$ maps symmetries into vectors composed of elements from $\ker D_y$, and the conservation law therefore takes the form
\begin{gather*} D_x \big( f \cdot \big(D_y + F_{u_x}^\top\big)\big(\mathcal{N}\big(\tilde{f}\big)\big) \big) = D_x (D_y(c[u])). \end{gather*}
The last equality means that $f \cdot (D_y + F_{u_x}^\top)(\mathcal{N}(\tilde{f})) = D_y(c[u]) + \bar{w}[u]$, where $\bar{w} \in \ker D_x$. Thus, \eqref{nocl} is the sum of a trivial conservation law and relations of the form $D_y (w)=0$, $D_x(\bar{w})=0$. Note that this is consistent with the work \cite{Sak}, in which the same structure for conservation laws of the Liouville equation $u_{xy}=e^u$ was derived from the Noether theorem.

\subsection{Systems that inherit Noether's operators due to dif\/ferential substitutions}\label{notr}
The way for construction of Noether's operators in the previous subsection is a special case of the following situation.
\begin{Lemma}\label{l2}
Let the relation $\tilde{Z} \circ P = Q \circ Z$ hold for differential operators $\tilde{Z}$, $Z$, $P$ and $Q$. Then $Q^\dagger \circ \tilde{\mathcal{N}} \circ P$ is an inverse Noether operator for $Z$ if $\tilde{Z}$ admits an inverse Noether opera\-tor~$\tilde{\mathcal{N}}$, and $P \circ \mathcal{N} \circ Q^\dagger$ is a Noether operator for $\tilde{Z}$ if $\mathcal{N}$ is a Noether operator for~$Z$.
\end{Lemma}
\begin{proof} If $\tilde{\mathcal{N}}$ is an inverse Noether operator for $\tilde{Z}$, then $ \tilde{Z}^\dagger \circ \tilde{\mathcal{N}} = \breve{\tilde{\mathcal{N}}} \circ \tilde{Z}$ by Def\/inition~\ref{nod}. Taking this and $Z^\dagger \circ Q^\dagger = P^\dagger \circ \tilde{Z}^\dagger$ into account, we obtain
\begin{gather*} Z^\dagger \circ Q^\dagger \circ \tilde{\mathcal{N}} \circ P = P^\dagger \circ \tilde{Z}^\dagger \circ \tilde{\mathcal{N}} \circ P = P^\dagger \circ \breve{\tilde{\mathcal{N}}} \circ \tilde{Z} \circ P = P^\dagger \circ \breve{\tilde{\mathcal{N}}} \circ Q \circ Z. \end{gather*}
Analogously, we can directly check that $\tilde{Z} \circ P \circ \mathcal{N} \circ Q^\dagger = Q \circ \breve{\mathcal{N}} \circ P^\dagger \circ \tilde{Z}^\dagger$ if $Z \circ \mathcal{N} = \breve{\mathcal{N}} \circ Z^\dagger$.
\end{proof}
It is easy to see that $D_x$, $D_y$ serve as $\tilde{Z}$ in \eqref{fintm} and $Z$ in \eqref{fsym}, while the identity mapping plays the role of Noether's operators $\tilde{\mathcal{N}}$ and $\mathcal{N}$ because $D_x$ and $D_y$ are skew-symmetric.

We can also apply the above lemma to situations when a dif\/ferential substitutions $v=\phi [u]$ maps system \eqref{hyp} into a system $E \{v\} = 0$, where $E \{v\}$ is an $\ell$-dimensional vector depending on $x$, $y$, $v$ and a f\/inite number of derivatives of $v$ (including, generally speaking, the mixed ones). Indeed, let $g\{\phi\}$ denote the function obtained from a function $g\{v\}$ by replacing all ${\partial^{i+j} v}/{\partial x^{i} y^{j}}$ with $D_x^i D_y^j\left(\phi [u]\right)$, and for dif\/ferential operators we let
\begin{gather}\label{dphi}
G\{\phi\}:=\sum _{i=0}^{k} \sum _{j=0}^{m} g_{ij}\{\phi\} D_x^i D_y^j \qquad \text{if} \qquad G\{v\} =\sum _{i=0}^{k} \sum _{j=0}^{m} g_{ij}\{v\} D_x^i D_y^j.
\end{gather}
Then $E\{\phi\}=0$ by the def\/inition of dif\/ferential substitutions. Linearizing $E\{\phi\}=0$ and taking~\eqref{dlin} into account, we obtain
\begin{gather}\label{dplin}
E_* \{\phi\} \circ \phi_* = Q \circ L,
\end{gather}
where $L$ is the linearization operator \eqref{lin} of system \eqref{hyp}, $Q$ is a dif\/ferential operator of the form~\eqref{do}, and
\begin{gather}\label{estar}
E_*\{v\} = \sum^{\infty }_{i=0} \sum^{\infty }_{j=0} {\frac{\partial E\{v\}}{\partial v_{i,j}}} D_x^i D_y^j, \qquad v_{i,j}:=\frac{\partial^{i+j} v}{\partial x^{i} y^{j}}.
\end{gather}
Applying Lemma~\ref{l2} to \eqref{dplin}, we obtain the following statement.
\begin{Corollary}\label{dp2n}
Let a differential substitutions $v=\phi [u]$ map all solutions of \eqref{hyp} into solutions of a system $E \{v\} = 0$. Then $\aleph = Q^\dagger \circ \tilde{\mathcal{N}} \circ \phi_*$ is an inverse Noether operator for system \eqref{hyp} if an operator $\tilde{\mathcal{N}}$ of the form \eqref{do} is an inverse Noether operator for $E_* \{\phi\}$, and $\tilde{\aleph} = \phi_* \circ \mathcal{N} \circ Q^\dagger$ is a Noether operator for $E_* \{\phi\}$ if $\mathcal{N}$ is a Noether operator for system~\eqref{hyp}. Here the differential operator $Q$ is defined by \eqref{dplin}, and $E_* \{\phi\}$ is obtained from the differential operator~\eqref{estar} by rule~\eqref{dphi}.
\end{Corollary}
\begin{Remark} Lemma~\ref{l2} and the above reasoning can also be applied to systems of a form that dif\/fers from \eqref{hyp}. For example, using the formula $(\partial_f (\phi))_*= \partial_f \circ \phi_* - \phi_* \circ (\partial_f -f_*)$ (see \cite{MSY}), we obtain \eqref{dplin} with $Q=\phi_*$ and $L=\partial_f -f_*$ if $v=\phi(x,u,u_x,\dots,u_m)$ map solutions of \eqref{evs} into solutions of a system $v_t=g(x,v,v_x,\dots,v_k)$. This gives us the formulas $\aleph = \left(\phi_*\right)^\dagger \circ \tilde{\mathcal{N}} \circ \phi_*$ and $\tilde{\aleph} = \phi_* \circ \mathcal{N} \circ \left(\phi_*\right)^\dagger$ for recalculating an inverse Noether operators $\tilde{\mathcal{N}}$ of $v_t=g(x,v,v_x,\dots,v_k)$ and a Noether operator $\mathcal{N}$ of~\eqref{evs} under the dif\/ferential substitution. These formulas are specialized versions of the formulas obtained in a dif\/ferent way for the case of B\"acklund transformations of evolution systems in \cite[equation~(24) and Theorem~4]{ff}.
\end{Remark}

Note that if $\tilde{\mathcal{N}}\{v\}$ is an inverse Noether operator for the system $E\{v\}=0$ (i.e., the coef\/f\/icients of $\tilde{\mathcal{N}}\{v\}$ are expressed only in terms of $x$, $y$, $v$ and derivatives of $v$, and the def\/ining relation
\begin{gather}\label{esvn}
\left(E_*\{v\}\right)^\dagger \circ \tilde{\mathcal{N}}\{v\} = \breve{\tilde{\mathcal{N}}}\{v\} \circ E_*\{v\}
\end{gather}
holds for any solution of $E\{v\}=0$), then the def\/ining relation \eqref{esvn} is true for solutions $v=\phi[u]$ too and, hence, $\tilde{\mathcal{N}}\{\phi\}$ is an inverse Noether operator for $E_*\{\phi\}$. Therefore, the system \eqref{hyp} admits the inverse Noether operator $\aleph = Q^\dagger \circ \tilde{\mathcal{N}}\{\phi\} \circ \phi_*$ by the above corollary. But for the operator $\tilde{\aleph}$ we have no guarantee that it can be expressed in terms of $x$, $y$, $v_{i,j}$ only and is a~Noether operator for $E\{v\}=0$ (i.e., not only for $E_* \{\phi\}$).

In addition, we cannot guarantee in the general situation that $\aleph$ and $\tilde{\aleph}$ are non-trivial Noether's operators, but can prove this in special cases. One of such special cases is formulated as follows.
\begin{Corollary}\label{coll2} Let $v(x,y)$ and $\tilde{F}(x,y,v,v_x,v_y)$ be $\ell$-dimensional vectors, $\ell \le n$, and let the system $v_{xy}=\tilde{F}(x,y,v,v_x,v_y)$ admit an inverse Noether operator $\tilde{\mathcal{N}}[v]$ of the form
\begin{gather*} \tilde{\mathcal{N}}[v] = \sum^r_{i=0} \eta_i[v] D_x^i + \sum^{\bar{r}}_{i=1} \bar{\eta}_i[v] D_y^i, \qquad r \ge 0, \, \bar{r}>0, \qquad \lambda[v]:=\det(\eta_r[v]) \ne 0, \end{gather*}
such that either $r>0$ or $\bar{\eta}_i=0$ for all $i$, and either $\operatorname{ord}_x(\lambda[v]) >0$ or $\operatorname{ord}_y(\lambda[v])=0$. If a diffe\-ren\-tial substitution $v=\phi[u]$ maps all solutions of \eqref{hyp} into solutions of $v_{xy}=\tilde{F}(x,y,v,v_x,v_y)$ and $\operatorname{ord}_x(\phi)=k>0$, $\operatorname{rank} (\phi_{u_k} ) = \ell$, then the relation
\begin{gather}\label{dplin2}
\big(D_x D_y - \tilde{F}_{v_x}[\phi] D_x - \tilde{F}_{v_y} [\phi] D_y - \tilde{F}_v [\phi] \big) \circ \phi_* = Q \circ L,
\end{gather}
holds for a differential operator $Q$ of the form \eqref{do}, and $\aleph = Q^\dagger \circ \tilde{\mathcal{N}}[\phi] \circ \phi_*$ is a non-trivial inverse Noether operator for system~\eqref{hyp}.

Here $L$ is defined by \eqref{lin}, $g[\phi]$ denotes the function obtained from a function $g[v]$ by replacing all ${\partial^i v}/{\partial x^i}$, ${\partial^j v}/{\partial y^j}$ with $D_x^i\left(\phi [u]\right)$, $D_y^j\left(\phi [u]\right)$, respectively, and $\tilde{\mathcal{N}}[\phi]$ is obtained from $\tilde{\mathcal{N}}[v]$ by rule \eqref{dphi}.
\end{Corollary}
\begin{proof} Due to Corollary~\ref{dp2n} and the comments near eq.~\eqref{esvn}, $\aleph$ is an inverse Noether operator for \eqref{hyp}. Therefore, we need only to prove the non-triviality of $\aleph$. Linearizing the def\/ining relation
\begin{gather*} D_x D_y(\phi) = \tilde{F} \left(x,y,\phi,D_x(\phi),D_y(\phi)\right) \end{gather*}
for the dif\/ferential substitution $v=\phi[u]$ and taking \eqref{dlin} into account, we see that $Q$ in \eqref{dplin2} has the form
\begin{gather*} Q = \sum^k_{i=0} \xi_i[u] D_x^i + \sum^m_{i=1} \bar{\xi}_i[u] D_y^i, \qquad m \ge 1.\end{gather*}
Comparing the coef\/f\/icients of $D_x^{k+1}D_y$ in the left- and right-hand sides of \eqref{dplin2}, we obtain $\xi_k=\phi_{u_k}$. This implies $\aleph = (-1)^k \phi_{u_k}^\top \eta_r[\phi] \phi_{u_k} D_x^{2k+r} + \cdots$, where the dots denote terms with lower powers of $D_x$. The conditions of the corollary guarantee that $\det (\eta_r[\phi]) \ne 0$ follows from $\det (\eta_r[v]) \ne 0$ and the coef\/f\/icient of $D_x^{2k+r}$ in $\aleph$ is not equal to zero.

Indeed, let $\operatorname{ord}_x (\lambda[v])$ be equal to $\delta$ and $\phi^i$ denote the $i$-th component of the vector $\phi$. Then the functions $D_x^{\delta}(\phi^i)$, $i=\overline{1,\ell}$, depend on $u_{\delta+k}$ and are functionally independent due to the condition $\operatorname{rank} (\phi_{u_k}) = \ell$ of the corollary and the equalities $(D_x^{\delta}(\phi^i))_{u_{\delta+k}} = \phi_{u_k}^i$. In addition, $D_x^{\delta}(\phi^i)$ cannot be expressed in terms of the functions $x$, $y$, $D_x^{j}(\phi)$, $0 \le j<\delta$, and, if $\delta >0$, $D_y^{b}(\phi)$, $b \in \mathbb{Z}_{+}$, because these functions do not depend on $u_{\delta+k}$. But, if $\lambda[\phi] = 0$, then a function $D_x^{\delta}(\phi^i)$ must be expressed in terms of $D_x^{\delta}(\phi^j)$, $j \ne i$, and the functions mentioned in the previous sentence. Therefore, $\det (\eta_r[\phi] ) = \lambda[\phi] \ne 0$ and $\phi_{u_k}^\top \eta_r[\phi] \phi_{u_k} \ne 0$.

Let $\theta = \sum\limits_{i=0}^p \sum\limits_{j=0}^s \zeta _{ij}[u] D_x^i D_y^j$. Then
\begin{gather*} \theta \circ L = \left( \sum_{j=0}^s \zeta _{pj}[u] D_y^j \circ \left(D_y - F_{u_x} \right) \right) \circ D_x^{p+1} + \cdots, \end{gather*}
where the dots denote terms with lower powers of $D_x$. Thus, the terms with the highest power of $D_x$ in $\theta \circ L$ must contain a non-zero power of $D_y$. But the highest power of $D_x$ in $\aleph$ does not contain $D_y$. Hence, $\aleph$ cannot be represented as $\theta \circ L$ and is non-trivial.
\end{proof}

As an illustrative example, let us consider the scalar Goursat equation $u_{xy}=\sqrt{u_x u_y}$. This equation is neither Euler--Lagrange nor Darboux integrable equation, but the dif\/ferential substitution $v=\sqrt{u_x}$ map it into the equation $v_{xy}= \frac{1}{4} v$, which has Lagrangian $v_x v_y + \frac{1}{4} v^2$ and the inverse Noether operator $\tilde{\mathcal{N}}=1$ (i.e., $\tilde{\mathcal{N}}$ is the identity mapping). Therefore, the Goursat equation admits an inverse Noether operator in accordance with Corollary~\ref{coll2}. Let us construct the corresponding operator $Q$. The direct calculations yield
\begin{gather*} \big( D_x D_y \big(\sqrt{u_x}\big)\big)_* = D_x \circ \big( D_y\big(\sqrt{u_x}\big) \big)_* - \frac{1}{4 \sqrt{u_y}} \circ L, \\
 \big( D_y\big(\sqrt{u_x}\big) \big)_* = D_y \circ \big(\sqrt{u_x}\big)_* - \frac{1}{2 \sqrt{u_x}} \circ L, \\
 \big( D_x D_y \big(\sqrt{u_x}\big)\big)_* = D_x \circ D_y \circ \big(\sqrt{u_x}\big)_* - Q \circ L, \qquad Q = \frac{1}{4 \sqrt{u_y}} + D_x \circ \frac{1}{2 \sqrt{u_x}},\\
\left( D_x D_y - \frac{1}{4} \right) \circ \big(\sqrt{u_x}\big)_* = \left( D_x D_y\big(\sqrt{u_x}\big) - \frac{\sqrt{u_x}}{4}\right)_* + Q \circ L = Q \circ L, \end{gather*}
and, according to Corollary~\ref{coll2},
\begin{gather*}\aleph= -8 Q^\dagger \circ \big(\sqrt{u_x}\big)_* = \left( \frac{2}{\sqrt{u_x}} D_x - \frac{1}{\sqrt{u_y}} \right) \circ \frac{1}{\sqrt{u_x}} D_x = \frac{2}{u_x} D_x^2 - \left( \frac{u_{xx}}{u_x^2} + \frac{1}{\sqrt{u_x u_y}} \right) D_x \end{gather*}
is an non-trivial inverse Noether operator for the Goursat equation. Since the substitution $v=\sqrt{u_y}$ also maps $u_{xy}=\sqrt{u_x u_y}$ into $v_{xy}= \frac{1}{4} v$, the `symmetric' (under the interchange $x \leftrightarrow y$) version of Corollary~\ref{coll2} gives us the second
inverse Noether operator
\begin{gather*} \frac{2}{u_y} D_y^2 - \left( \frac{u_{yy}}{u_y^2} + \frac{1}{\sqrt{u_x u_y}} \right) D_y \end{gather*}
for the Goursat equation.

In conclusion, we note that the situations considered in this section are fairly typical for hyperbolic integrable systems. For example, all $S$-integrable scalar equation \eqref{hyp} mentioned in \cite{MeshSok,ZhibSok} are either Euler--Lagrange equations or mapped into Euler--Lagrange equations via dif\/ferential substitutions. Therefore, all these $S$-integrable equations admit at least inverse non-trivial Noether operators by Corollary~\ref{coll2}.
In addition, most of the other scalar equations admitting higher symmetries and listed in these works are Darboux integrable, and the remaining equations from these works are related to the Euler--Lagrange equation $v_{xy}= c v$ via B\"acklund transformations, some of which have the form of the dif\/ferential substitutions $v=\phi[u]$. Thus, the vast majority of these equations admit non-trivial Noether's operators. The full lists of scalar equations~\eqref{hyp} related via dif\/ferential substitutions of f\/irst order to Euler--Lagrange equations $v_{xy}=G(v)$ can be found in \cite{ZhPP}. The introduction of \cite{Slin} mentions a constructive (but weak) necessary condition for a scalar equation~\eqref{hyp} to be mapped into an equation of the form $v_{xy}=G(x,y,v)$ by a dif\/ferential substitution of higher order.

\subsection*{Acknowledgements} The author thanks the referees for useful suggestions. This work is supported by the Russian Science Foundation (grant number 15-11-20007).

\pdfbookmark[1]{References}{ref}
\LastPageEnding

\end{document}